\documentclass[
  twocolumn,
  aps,
  prl,
  showpacs,
  preprintnumbers,
  letterpaper
]{revtex4}

\usepackage{inputenc}
\usepackage[T1]{fontenc}
\usepackage{lmodern}
\usepackage{amsmath,amsfonts,amssymb,amstext,amscd,amsthm,dsfont}
\usepackage[mathscr]{eucal}
\usepackage{mathtools}
\usepackage{afterpage}
\usepackage{calrsfs,enumerate}
\usepackage[normalem]{ulem}
\usepackage{graphicx}
\usepackage{latexsym}
\usepackage{bm}
\usepackage{makeidx}
\usepackage{subfigure}
\usepackage{bbold}
\usepackage{color}

\newcommand{\mes}{|\psi^+\rangle\langle\psi^+|}
\newcommand{\doi}{\Delta\otimes\id}
\newcommand\ket[1]{\left|#1\right>}
\newcommand\bra[1]{\left<#1\right|}
\newcommand\braket[2]{\left<#1\right|\left.\!\!#2\right>}
\newcommand\projj[2]{\left|#1\right>\!\left<#2\right|}
\newcommand\proj[1]{\left|#1\right>\!\left<#1\right|}

\DeclareMathOperator{\Tr}{\operatorname{Tr}}
\newcommand{\deff}{\vcentcolon=}

\newcommand{\rab}{\rho_{AB}}
\newcommand{\sab}{\sigma_{AB}}

\newcommand{\ra}{\rho_{A}}

\newcommand{\Lam}{\Lambda}
\newcommand{\id}{\textup{id}}
\newcommand\bnorm[2]{\left\|{#1}\right\|_{#2}}
\newcommand{\llangle}{\langle\langle}
\newcommand{\rrangle}{\rangle\rangle}

\newcommand{\I}{\openone}

\newcommand{\SR}{\textup{SR}}
\newcommand{\OSR}{\textup{OSR}}

\newcommand{\CDP}{\textup{CDP}}

\newcommand{\reff}[1]{Eq.~\eqref{#1}}

\newcommand*{\cD}{\mathcal{D}}

\newcommand*{\cH}{\mathcal{H}}

\newcommand*{\cL}{\mathcal{L}}

\newcommand*{\cT}{\mathcal{T}}

\newtheorem{Proposition}{Proposition}

\newtheorem{Theorem}{Theorem}

\newtheorem{Lemma}{Lemma}

\begin{document}

\title{Channel discrimination power of bipartite quantum states}

\author{ Matteo Caiaffa, Marco Piani}
\affiliation{SUPA and Department of Physics, University of Strathclyde, Glasgow G4 0NG, UK}

\begin{abstract}
We quantify the usefulness of a bipartite quantum state in the ancilla-assisted channel discrimination of arbitrary quantum channels, formally defining a worst-case-scenario channel discrimination power for bipartite quantum states. We show that such a quantifier is deeply connected with the operator Schmidt decomposition of the state. We compute the channel discrimination power exactly for pure states, and provide upper and lower bounds for general mixed states. We show that highly entangled states can outperform any state that passes the realignment criterion for separability. Furthermore, while also unentangled states can be used in ancilla-assisted channel discrimination, we show that the channel discrimination power of a state is bounded by its quantum discord.
\end{abstract}

\maketitle

A quantum channel is the most general linear transformation a quantum system can undergo, capturing mathematically the notion of physical process and playing the role of  basic block in quantum information processing~\cite{nielsen2010quantum}. A fundamental task that falls under the umbrella of quantum metrology~\cite{giovannetti2004quantum,toth2014quantum} is channel discrimination~\cite{DArianoPP01}.
Channel discrimination is the task of telling apart two or more \emph{known} channels which are each applied with some random a priori distribution to an input of our choice; think of the situation where we want to probe the presence or absence of a known magnetic field. 
In the prototypical and simplest case,  two channels are applied just once with equal a priori probability distribution, and we perform a measurement on the output probe trying to infer the which-channel information. The goal is that of identifying the channel applied, with the highest possible probability of success. 
Channel discrimination is typically performed by tailoring the state of the input probe to the channels to be discriminated. In general, the wrong choice of input might not only make the probability of correct identification less than optimal, but it might make it altogether impossible, in the sense that, for some choice of input state, the output state could be the same for both channels, even when the latter differ.

There can be advantages in channel discrimination by making use of correlations between the  probe and a reference ancilla. One possible advantage is that correlations may lead to a probability of success in the discrimination that is higher than what possible without the use of an ancillary system
~\cite{DArianoPP01,kitaev1997quantum,childs2000quantum,d2001using,
acin2001statistical,giovannetti2004quantum,gilchrist2005distance,
rosgen2005hardness,sacchi2005optimal,sacchi2005entanglement,lloyd2008enhanced,
rosgen2008additivity,watrous,watrous2008distinguishing}. 
In general, achieving such a higher probability of success requires (i) to tailor the probe-ancilla input state to the specific channels to be discriminated and (ii) input entanglement between probe and ancilla. Another advantage provided by probe-ancilla correlations, on which we focus in this Letter, is that they may allow to discriminate between an arbitrary pair of known channels, without the need to tailor the input probe-ancilla state to avoid `being blind' to the difference between the channels. This fact is at the basis of the celebrated Choi-Jamio{\l}kowski isomorphism~\cite{choi1975completely,jamiolkowski1972linear} between linear maps and linear operators, and allows to perform channel tomography -- that is, to identify an unknown channel with many uses of the unknown channel -- with a fixed input probe-ancilla state~\cite{altepeter2003ancilla}. Such a feat can be achieved even in the absence of entanglement, and Ref.~\cite{altepeter2003ancilla} already identified the Operator Schmidt Rank (OSR; to be defined later) of the probe-ancilla input state as the key property determining whether such state makes ancilla-assisted tomography possible. Nonetheless, the study of the usefulness of correlations in fixed-input ancilla-assisted channel discrimination and channel tomography has been limited~\cite{jenvcova2016conditions}. In this Letter, we shed light on ancilla-assisted channel discrimination, providing an analysis of how the Operator Schmidt Decomposition (OSD; to be defined later) of the probe-ancilla input state affects the quality of the discrimination. In particular, we introduce a worst-case quantifier for the performance of a probe-ancilla state in channel discrimination, the \emph{Channel Discrimination Power} (CDP), and we provide general upper and lower bounds to it in terms of the OSD of the state. We compute the exact CDP of pure states.  Remarkably, we show that, while correlated but unentangled states can have non-zero CDP, and allow the discrimination of any pair of channels as long as they have maximal OSR, they cannot have maximal CDP. More in general, we provide a non-trivial bound on the channel discrimination power of any state -- entangled or unentangled -- that passes the so-called realignment (or computable cross-norm) criterion for separability~\cite{chen2003matrix,rudolph2004computable}. Furthermore, we prove that the general quantumness of correlations known as quantum discord~\cite{RevModPhys.84.1655} provides a bound for the channel discrimination power of a bipartite state.

\emph{Notation and preliminaries.}
We will limit ourselves to finite-dimensional systems. Hence, each Hilbert space $\cH$ will be equivalent to $\mathbb{C}^d$ for some integer dimension $d$. The space of linear operators $L$ (equivalent to matrices) on $\cH$ will be indicated by $\cL(\cH)$. 
We will be interested in the $p$-norms
$\|L\|_p
:=
\Big(
	\Tr
		\big(
			(L^\dagger L)^{\frac{p}{2}}
		\big)
\Big)^\frac{1}{p},
$
for the values $p=1,2,\infty$~\cite{horn2012matrix}. 
 
We indicate by $d_X$ the dimension of a system $X$ with Hilbert space $\cH_X$. We will focus on bipartite systems $AB$, and, unless stated otherwise, we will define $d_{\min} = \min\{d_A,d_B\}$. A quantum state on $\cH$ corresponds to a density operator $\rho$ belonging to convex subset $\cD(\cH)\subset \cL(\cL)$ of operators that have unit trace and are positive semidefinite. We indicate by $\Tr$ the trace operation, by $\Tr_X$ the partial trace on system $X$.
We denote by $\rho_X$ the (reduced) state of system $X$. The space $\cL(\cH)$ can be made into a Hilbert space itself by considering the Hilbert-Schmidt inner product $\llangle C | D \rrangle := \Tr(C^\dagger D)$ between two operators $C,D\in L(\cH)$, which induces the 2-norm via $\|C\|_2 = \sqrt{\llangle C | C \rrangle}$.

The trace distance between two density matrices $\rho$ and $\sigma$ is defined as $D(\rho,\sigma):=\frac{1}{2}\|\rho-\sigma\|_1$~\cite{nielsen2010quantum}. 
Its operational meaning is that of bias in the optimal discrimination of the two states. More specifically, the probability of correctly identifying the state of a system that is a priori in the state $\rho$ or $\sigma$ each with $50\%$ chance, in the single-shot scenario when one is given one copy of the state to measure, is $\left(1+D(\rho,\sigma)\right)/2$. The trace distance $D(\rho,\sigma)$ varies between $0$ (for identical states) to $1$ (for perfectly distinguishable states, which are mathematically orthogonal, $\llangle \rho | \sigma \rrangle = 0$).
A bipartite state $\rho_{AB}$ is unentangled (or separable)
if it is the convex combination of product (or uncorrelated) states~\cite{revent},
\begin{equation}
\label{eq:separable}
\rho_{AB} = \sum_i p_i \rho^A_i \otimes \rho^B_i
\end{equation}
with $\{p_i\}$ a probability distribution. That is, a state is unentangled if all the correlations that the systems $AB$ exhibit have an explanation in terms of shared classical randomness. A state is entangled if it is not separable. To decide whether a given state is separable or not (that is, whether it admits a decomposition like \eqref{eq:separable} or not) is a hard problem in general, but numerous criteria have been devised to tackle it~\cite{revent}.

The physical evolution of a quantum system is formally described in terms of quantum channels~\cite{nielsen2010quantum}. In general, one considers evolutions from an input system $X$ to an output system $Y$, representing evolution in time or general transfer of information -- either in space or in time -- from one system to another. Formally, a quantum channel from $X$ to $Y$ is a completely-positive trace-preserving linear map $\Lambda$ from $\cL(\cH_X)$ to $\cL(\cH_Y)$.

\emph{Operator Schmidt Decomposition.}
Any vector state $\ket{\psi}_{AB}\in \cH_A\otimes \cH_B$ admits a Schmidt decomposition~\cite{nielsen2010quantum}
\begin{equation}
\label{eq:SD}
\ket{\psi}_{AB} = \sum_{i=1}^{\SR(\psi)} \sqrt{p_i}\ket{a_i}_A\otimes\ket{b_i}_B,
\end{equation}
with $\{\sqrt{p_i}\}$ positive numbers that satisfy $\sum_{i=1}^{\SR(\psi)} (\sqrt{p_i})^2 = \sum_{i=1}^{\SR(\psi)}p_i = \braket{\psi}{\psi}$. Since $\braket{\psi}{\psi}=1$, we can think of $\{p_i\}$ as of a probability distribution, whose elements we can imagine ordered, $p_1\geq p_2 \geq\ldots$, without loss of generality. Furthermore $\{\ket{a_i}\}$ and $\{\ket{b_i}\}$ are some special and $\ket{\psi}$-dependent orthonormal bases for $\cH_A$ and $\cH_B$, respectively. Here $\SR(\psi)$ denotes the Schmidt rank of $\ket{\psi}_{AB}$, which is the number of non-zero $p_i$'s, and satisfies $\SR(\psi)\leq d_{\min}$.
Let $\rho_{AB}$ be a density matrix for the bipartite system $AB$. We can consider it as a vector in $\cL(\cH_A\otimes\cH_B)$, and hence derive the Operator Schmidt Decomposition (see~\cite{aniello2009relation,lupo2008bipartite} and references therein)
\begin{equation}
\label{eq:OSD}
\rho_{AB} = \sum_{i=1}^{\OSR(\rho)}r_i A_i \otimes B_i.
\end{equation}
Here $OSR(\rho)$ is the number of non-zero Operator Schmidt Coefficients (OSCs) $r_i$, and $\{A_i\}_{i=1}^{d_A^2}$ and $\{B_i\}_{i=1}^{d_B^2}$ are some ($\rho$-dependent) orthonormal bases for the spaces $\cL(\cH_A)$ and $\cL(\cH_B)$, respectively. The OSR is the minimum number of product terms that need to enter in any decomposition of $\rab$. Since $\rho_{AB}$ is Hermitian, one can argue that the two orthonormal operator bases in~\eqref{eq:OSD} can be (but need not be) taken to be composed of Hermitian operators. The OSCs are the singular values of the correlation matrix $[C_{ij}(\rho_{AB})]_{ij}$, with $C_{ij}(\rho_{AB}):=\llangle F_i\otimes G_j | \rho_{AB}\rrangle$, where $\{F_i\}$ and $\{G_j\}$ are arbitrary local orthonormal bases for operators. We will take the OSC to be ordered as $r_1 \geq r_2 \geq \ldots $;
they satisfy $\sum_ir_i^2 = \llangle \rho | \rho \rrangle = \Tr(\rho^2)$. Notice that $\OSR(\rho_{AB})\leq d_{\min}^2$, as the vector space $L(\cH_A)$ has dimension $d_A^2$ (similarly for $L(\cH_B)$). It is immediate to realize that the SD of a pure state $\ket{\psi}_{AB}$ and the OSD of the corresponding density matrix $\proj{\psi}_{AB}$ are related: indeed, for a pure state, $r_i=\sqrt{p_k}\sqrt{p_l}$, $A_i = \projj{a_k}{a_l}$, and $B_i = \projj{b_k}{b_l}$, for $i=(k,l)$ a multi-index.

A powerful criterion of separability is the computable cross-norm (or realignment) criterion~\cite{chen2003matrix,rudolph2004computable}, which states that, if $\rho_{AB}$ is unentangled, then its OSCs satisfy $\sum_i r_i \leq 1$. Thus, if one finds $\sum_i r_i > 1$, one can conclude that $\rho_{AB}$ is entangled.

\emph{Channel discrimination and channel tomography.}
Channel discrimination is a generalization of state discrimination, where the objects to tell apart are now channels. One can define a physically meaningful notion of distance between two channels  $\Lambda_0$ and $\Lambda_1$ via~\cite{watrous}
\begin{equation}
D(\Lambda_0,\Lambda_1):=\max_{\rho\in \cD(\cH_S)} D\left(\Lambda_0[\rho],\Lambda_1[\rho]\right),
\end{equation}
that is by considering the trace distance of the output states of a probe upon acting on the same input state of the probe. One fundamental---and relevant for applications---way in which quantum physics differs from classical physics, is that the distinguishability of two channels, as captured by $D(\Lambda_0,\Lambda_1)$, can be enhanced by the use of entanglement between the input probe and an ancilla~\cite{DArianoPP01,kitaev1997quantum,childs2000quantum,d2001using,
acin2001statistical,giovannetti2004quantum,gilchrist2005distance,
rosgen2005hardness,sacchi2005optimal,sacchi2005entanglement,lloyd2008enhanced,
rosgen2008additivity,watrous,watrous2008distinguishing}. One can prove that the best ancilla system can be chosen to be a copy $S'$ of the input probe system $S$, so that we can define the so-called diamond distance between $\Lambda_0$ and $\Lambda_1$ as
\begin{equation}
\label{eq:diamond}
D_\diamond(\Lambda_0,\Lambda_1):=D(\Lambda_{0,S}\otimes\id_{S'},\Lambda_{1,S}\otimes\id_{S'}),
\end{equation}
where $\id_X$ indicates the identity map on system $X$.
The diamond distance formalizes the notion of best possible one-shot distinguishability between two quantum channels. 

In general, it is not possible to distinguish arbitrary quantum channels in $\cT(\cH_X,\cH_Y)$ by means of their action on an input state $\rho\in\mathcal{D}(\cH_X)$ of the probe alone that is independent of the channels considered, as there are always two different channels that have the same effect on a given input state~\footnote{For example, consider the case where $\Lambda_0$ is the identity channel, so that $\Lambda_0[\sigma]=\sigma$ for all $\sigma$, and $\Lambda_1$ is the channel with fixed output $\rho$. Then, obviously, $D(\Lambda_0[\rho],\Lambda_1[\rho]) = 0$, even if the two channels are very different, and even having many copies of $\Lambda_i[\rho]$ we cannot tell the two channels apart.}. 
Nonetheless, it is always possible to tell two arbitrary channels in $\cT(\cH_X,\cH_Y)$ apart by `feeding' them with many different input states $\rho_k$. Indeed, as long as $\{\rho_k\}$ constitutes a basis for $\cL(\cH_X)$, and as long an arbitrary number of uses of the channel are allowed, one can perform a tomographic reconstruction of a channel $\Lambda$~\cite{nielsen2010quantum}, even in the case where there is no prior information at disposal about the channel (see also Figure~\ref{fig:NE}).

Remarkably, it is possible to perform tomography of the channel, or the non-trivial discrimination of an arbitrary number of channels, even with just a fixed input state, as long as one uses an ancilla: this constitutes the framework of ancilla-assisted channel discrimination and channel tomography (see Figure~\ref{fig:diamond}). Ref.~\cite{altepeter2003ancilla} proves both theoretically and experimentally that channel tomography is possible also when the state $\rho_{AB}$ of probe $A$ and ancilla $B$ is separable. The key condition that permits channel tomography on $A$ with $\rho_{AB}$ is that $\OSR(\rho_{AB})=d_{A}^2$. Indeed, one has
\[
\Lambda_A[\rab]=\sum_{i=1}^{\OSR(\rho)}r_i  \Lambda[A_i] \otimes B_i,
\]
and, as long as the state has $\OSR(\rho)=d_A^2$, one can reconstruct the action of the map $\Lambda$ on an arbitrary state $\sigma\in D(\cH_A)$ as follows:
\[
\Lambda[\sigma]=\sum_{i=1}^{d_A^2} \frac{1}{r_i}\llangle A_i | \sigma \rrangle \Tr_B(\openone_A \otimes B_{i,B}^\dagger \Lambda_A[\rab]).
\]
We improve on this basic observation, by introducing and studying a simple and meaningful measure of merit for the usefulness of a fixed probe-ancilla state in channel discrimination.

\emph{Channel discrimination power.}
For any quantum state $\rab\in \cD(\cH_A\otimes\cH_B)$, we define the channel discrimination power (CDP) of $\rab$ on $A$ as
\begin{equation}\label{cdp}
\CDP_A(\rab)
:=
\inf_{\Lam_0,\Lam_1}
\frac{
		D(\Lam_{0,A}[\rho_{AB}],
			\Lam_{1,A}[\rho_{AB}])
	}
	{
		D_\diamond(\Lam_0,\Lam_1).
}
\end{equation}
The infimum is taken over all pairs $\Lambda_0,\Lambda_1$ of quantum channels with input in $\cL(\cH_A)$, and we have used the short-hand notation $\Lam_{0,A}:=\Lam_{0,A}\otimes \id_B$. We similarly define $\CDP_B(\rab)$. The parameter $\CDP_A(\rab)$ captures how suitable $\rab$ is for ancilla-assisted channel discrimination, comparing how well $\rab$ allows us to discriminate two channels acting on $A$ with respect to the optimal distinguishability of those two channels, in a worst-case scenario approach. Notice that we must necessarily  take into account the actual distiguishability of the two channels $\Lambda_0$ and $\Lambda_1$; a minimization of the numerator alone in Eq.~\eqref{cdp} would trivially vanish. In principle one could consider another measure of distinguishability of the two channels to be used as denominator, 
like $D(\Lambda_0,\Lambda_1)$, but the diamond distance is a choice that is mathematically convenient and conceptually meaningful, as it  regards the usefulness of an ancilla in the discrimination. We do not know whether the infimum in Eq.~\eqref{cdp} can be replaced by a minimum, more specifically, whether the infimum can be achieved with some bounded output dimension for the channels $\Lambda_0,\Lambda_1$. 
In the following we report a number of results about the channel discrimination power~\cite{appendix}.

\emph{Basic properties.} One can easily prove that $\CDP_A(\rho_{AB})$ is continuous in its argument: 
\begin{equation}
\label{eq:cont}
|\CDP_A(\rab)-\CDP(\sigma_{AB})|\leq 2 D(\rab,\sigma_{AB}).
\end{equation}
Furthermore, it is monotonically non-increasing under local operations on the ancilla, that is, $\CDP_A(\rab))\geq \CDP_A(\Gamma_B[\rab])$, for all channels $\Gamma$ on $B$.
Notice that this immediately implies that, for fixed dimension of $A$, the $\CDP$ assumes maximal value for pure states, as any bipartite state $\rho_{AB}$ can be seen as the reduced state of a pure state $\psi_{ABB'}$, with $B'$ a purifying system, and $BB'$ considered together as one ancilla. Furthermore, $\CDP_A(\rho_{AB})$ is invariant under local unitaries on $A$, that is $\CDP_A(\rab))= \CDP_A(U_A\rab U_A^\dagger)$. Together with monotonicity under operations on $B$, this implies that the $\CDP$ of a pure state only depends on its Schmidt coefficients. We find:
\begin{Theorem}
	 Let $\ket{\psi}_{AB}$ be a pure state with Schmidt decomposition as in \eqref{eq:SD}. Then, if $d_{\min} = d_A = d_B$, $\CDP_A(\psi_{AB})=\CDP_B(\psi_{AB})=p_{d_{\min}}$, while, if $d_{\min} = d_A < d_B$, $\CDP_A(\psi_{AB})=p_{d_{\min}}$ and $\CDP_B(\psi_{AB}) = 0$.
\end{Theorem}
Notice that it might be that $p_{d_{\min}}=0$, in which case both $\CDP_A(\psi_{AB})$ and $\CDP_B(\psi_{AB})$ vanish. We remark that $p_{d_{\min}}$ is a quantifier of the entanglement of $\ket{\psi}_{AB}$. Having already established that $\CDP_A$ is maximal for pure states, we find that it achieves its maximum, $1/d_A$, for maximally entangled states, e.g., for $\ket{\psi^+}_{AB}=\frac{1}{\sqrt{d_A}}\sum_{i=1}^{d_A} \ket{i}_A\ket{i}_B$.

We remark that it should not be surprising that the maximum of the channel discrimination power, being defined as in Eq.~\eqref{cdp}, decreases with $d_A$, since the number of parameters describing an arbitrary channel with input in $A$ increases with the size of $A$.

\emph{General bounds for mixed states.}
We now present general bounds for the CDP.

\begin{Theorem}\label{16}
	Let $\rab$ have an OSD as in Eq.~\eqref{eq:OSD}, 
	with $\{A_i\}$, $\{B_i\}$ Hermitian orthonormal bases for $\cL(\cH_A)$ and $\cL(\cH_B)$, respectively. Then 
	\begin{equation}
	\label{eq:CDPbounds}
	\frac{r_{d_A^2}}{d^{5/2}_A} \leq \CDP_A(\rab)\leq \min_i\left\{r_i\frac{\bnorm{B_i}{1}}{\bnorm{A_i}{\infty}}\right\} \leq r_{d_A^2}\sqrt{d_Ad_B}.
	\end{equation}
\end{Theorem}

The last inequality comes from standard dimension-dependent relations between $p$-norms~\cite{horn2012matrix}. These bounds are not tight in general, as proven by the results about pure states. Nonetheless, they capture quantitatively, rather than purely qualitatively, the fact that the necessary and sufficient condition for $\rab$ to always enable ancilla-assisted discrimination and tomography of an arbitrary channel with input in $\cL(\cH_A)$ is that $\OSR(\rho)=d_A^2$.

\emph{Bound for separable states.} 
We recall that mixed unentangled states may have maximal OSR, that is $\OSR(\rho_{AB})=d_A^2$, so that, according to Eq.~\eqref{eq:CDPbounds}, they have non-zero CDP. This is the case, for example, of isotropic states, considered more in detail below.

We now focus on the case $d_A=d_B=d$. As we have seen, $\CDP$ can be as high as $1/d$. We prove that such a value cannot be achieved by states passing the realignment criterion for separability, i.e., such that its OSCs satisfy $\sum_ir_i\leq 1$. The proof makes use of the following bound, which characterizes the total correlations present in a state, and may be of independent interest.
\begin{Lemma}
For any $\rho_{AB}$ and any product state $\sigma_A\otimes \sigma_B$, one has
$\sum_{i\geq2}r^2_i(\rho_{AB}) = \Tr(\rho^2)- r_1^2 \leq \|\rho_{AB} - \sigma_A\otimes \sigma_B\|^2_2$.
\end{Lemma}
Such lemma allows us to prove the following.
\begin{Theorem}
	If the OSCs of $\rab$ satisfy $\sum_ir_i\leq 1$, then $r_{d}\leq r_{\textup{CN}}$ with
	\[
	r_{\textup{CN}} = \frac{d(d^2-1)-\sqrt{d^2-1}}{d(d^2-1)^2+d^3}<\frac{1}{d^2}.
	\]
\end{Theorem}

By combining this with Theorem~\ref{16} we prove that, if the OSCs of $\rab$ satisfy $\sum_i r_i \leq 1$, then $\CDP(\rab)\leq r_{\textup{CN}} d < 1/d$. We remark that the realignment criterion for separability is satisfied by all separable states, and by many (weakly) entangled states~\cite{rudolph2004computable,chen2003matrix,revent}.

\emph{Relation with discord.}
As we have just seen, entanglement is needed to achieve the maximal possible CDP. Nonetheless, separable states can have non-vanishing CDP, 
when they have maximal OSR. As pointed out in Ref.~\cite{dakic2010necessary}, this is not possible for states that 
do not exhibit quantum discord. A bipartite state is classical on $A$ if it can be expressed as $\rab = \sum_i p_i \proj{a_i}_A\otimes \rho_i^B$, for some orthonormal basis $\{\ket{a_i}\}$, and manifestly has $\OSR\leq d_A$. States that are not classical on $A$ are said to possess quantum discord~\cite{henderson2001classical,ollivier2001quantum,RevModPhys.84.1655} and may be detected as discordant by looking at their OSR~\cite{dakic2010necessary,lanyon2013experimental}. All entangled states necessarily possess discord, but also unentangled states can. Discord plays a basic role in quantum information processing, being linked to the impossibility of local broadcasting of correlations and information~\cite{piani2008no},
 to quantum data hiding~\cite{piani2014quantumness}, to quantum data locking~\cite{boixo2011quantum}, to entanglement distribution~\cite{chuan2012quantum,streltsov2012quantum}, to quantum metrology~\cite{girolami2014quantum}, to quantum cryptography~\cite{pirandola2014quantum}. Here we shed light on the role of discord in the latter. By using the continuity~\eqref{eq:cont}, we find that
\[
\begin{split}
\CDP_A(\rab)
&\leq 
\min_{
	\substack{
		\Lam_A \text{\ s.t.} \\ OSR\left(\Lam_A[\rab]\right)\,<\,
		d_A^2}
	}2D(\rab,\Lam_A[\rab])  \\
&\leq \min_{\Pi_A}2D(\rab,\Pi_A[\rab]).
\end{split}
\]
The first minimization is over channels that reduce the OSR of $\rab$ to less than maximal. The second minimization is over projective measurements $\Pi[L]=\sum_i\proj{a_i}L\proj{a_i}$, for a choice of basis $\{\ket{a_i}\}$ to be optimized over. The quantity on the second line is a known geometric discord quantifier~\cite{luo2008using}. Hence, we have found that the bipartite state $\rab$ must be contain a large amount of discord in order for $\rab$ to be useful in one-shot, worst-case ancilla-assisted channel discrimination. 

\begin{figure}[!t]
	\subfigure[]{
		\includegraphics[width=.5\linewidth]{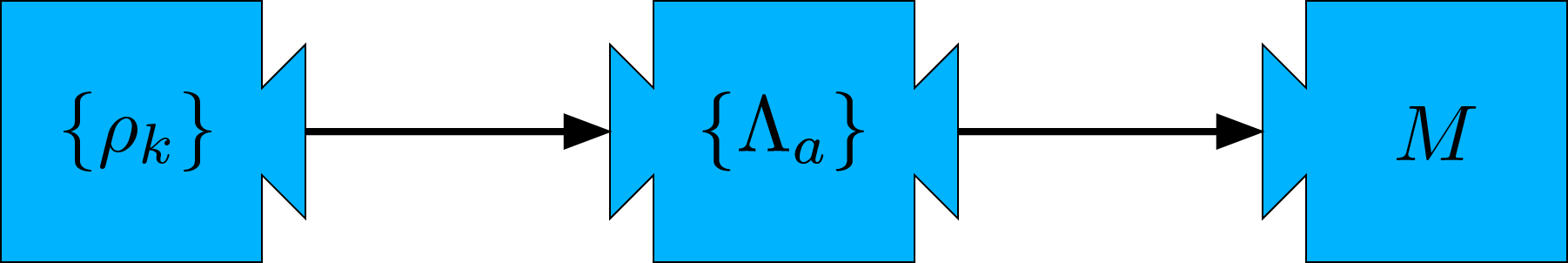}
		\label{fig:NE}
	}
	\vspace{1mm}
	\subfigure[]{
		\includegraphics[width=0.5\linewidth]{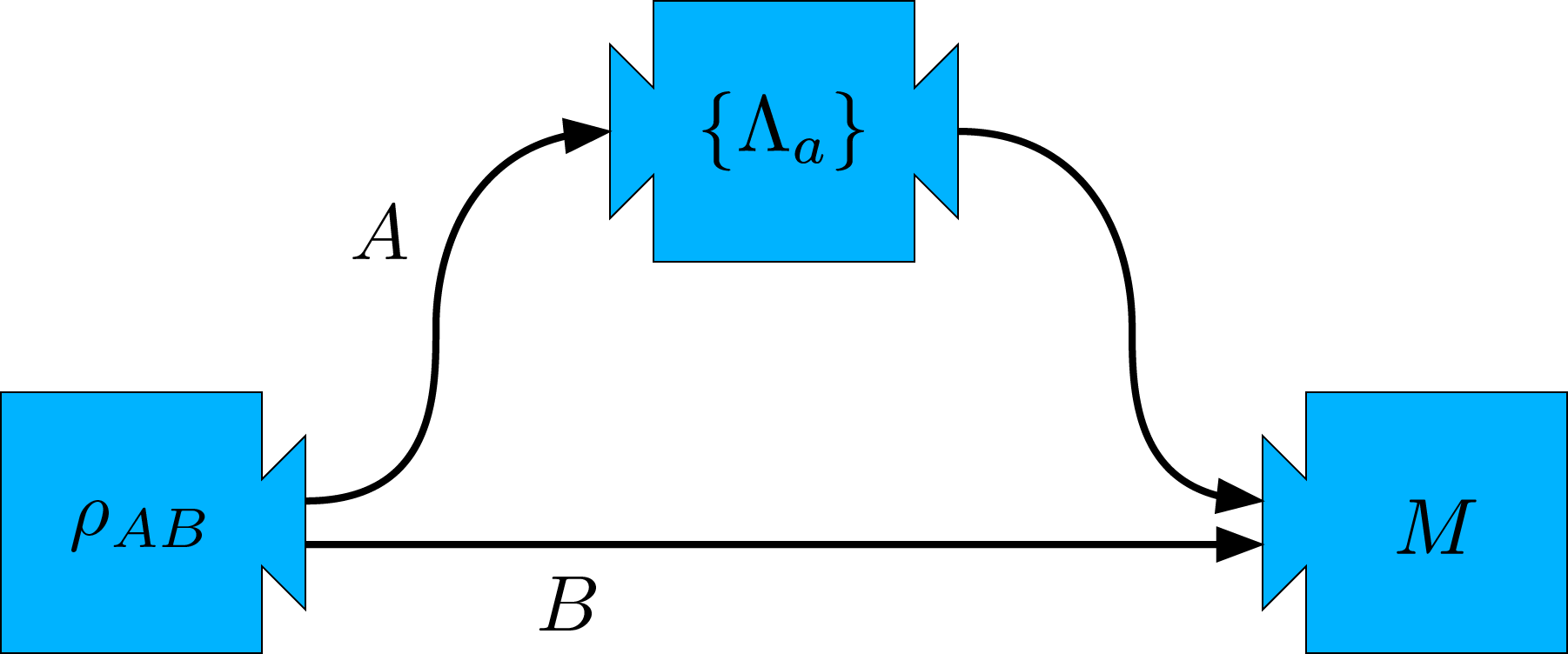}
		\label{fig:diamond}
	}
	\caption{Two strategies for distinguishing channels. (a) No
		ancilla is used: a probe undergoes one of many possible 
		quantum evolutions described by channels $\{\Lambda_a\}$, and is later measured (box $M$). Many different input states $\{\rho_k\}$ are in general needed to achieve the ability to discriminate between arbitrary channels, especially if one cannot tailor the input to the channels. (b) Ancilla-assisted: the probe $B$ is correlated with an ancilla $A$;
		the output probe and the ancilla are jointly measured. Depending on the initial probe-ancilla correlations, it might be possible to distinguish between arbitrary evolutions, without modifying the input.}
	\label{fig:fig}
\end{figure}

\emph{Example.} As an example that goes beyond pure states, we consider the class of isotropic states, i.e., states of the form~\cite{horodecki1999reduction}
\begin{equation}\label{eq:isostates}
\rho_{\text{iso}}(p)
=(1-p)\frac{\openone_{AB}}{d^2}+p\proj{\psi^+}_{AB}.
\end{equation}
This is a paradigmatic class of noisy states that interpolates between an uncorrelated state (for $p=0$) and a maximally entangled state (for $p=1$). It is known that isotropic states are separable for $0\leq p\leq \frac{1}{d+1}$ and entangled for $\frac{1}{d+1}<p\leq 1$. This is also the class of states used in Ref. \cite{altepeter2003ancilla} in the context of ancilla-assisted channel tomography, where it was already observed that this class of states enables channel tomography as soon as $p>0$. Indeed,
one checks easily that isotropic states have the OSD
\begin{equation} 
\rho_{\text{iso}}(p)=\frac{1}{d} \frac{\openone}{\sqrt{d}} \otimes\frac{\openone}{\sqrt{d}}
+
\frac{p}{d}\sum_{k=2}^{d^2} A_k\otimes A^*_k,
\end{equation}
where $\{A_k\}$ is any collection of $d^2-1$ traceless orthonormal operators, and complex conjugation is taken in the local Schmidt basis of the maximally entangled states. Thus, the OSCs of $\rho_{\text{iso}}(p)$ are evidently $(1/d,p/d,\ldots,p/d)$. Notice that $r_{d^2} = p/d$, so that the general bounds \eqref{eq:CDPbounds} become $p/d^{7/2} \leq \CDP(\rho_{\text{iso}}(p))\leq p$; we are able to prove the bounds
\begin{equation}
\frac{p}{d+1-p}
\leq
\CDP(\rho_{\text{iso}}(p))
\leq\min\left\{\frac{2p}{d},\frac{1}{d}\right\},
\end{equation}
which reproduce the correct value for $\CDP$ in the limit $p\rightarrow 1$ in which the isotropic states become maximally entangled.

\emph{Conclusions.} Quantum correlations~\cite{revent,RevModPhys.84.1655} play an important role in several areas of physics, going from quantum foundations, to quantum condensed-matter physics, to quantum information processing and quantum technologies. In particular, quantum correlations can be exploited in quantum metrology~\cite{giovannetti2004quantum,toth2014quantum}.
In this Letter we have focused on the usefulness of quantum correlations for ancilla-assisted channel discrimination with fixed input, introducing a quantifier of such usefulness: the channel discrimination power (CDP) of the state. We have argued that the key relevant parameter that dictates the CDP of a state is its smallest operator Schmidt coefficient. We have proven that the CDP is maximal for maximally entangled states. This can considered an argument to consider the Choi-Jamio{\l}kowski isomorphism~\cite{choi1975completely,jamiolkowski1972linear} as the best possible one-to-one mapping between states and maps. The general bounds for the CDP that we derived allowed us to prove that, while also unentangled states permit ancilla-assisted fixed-input channel discrimination and channel tomography, highly entangled states outperform---in the sense of having a larger CDP---all states that pass the so-called realignment criterion of separability~\cite{chen2003matrix,rudolph2004computable}. We also add to  the list of quantum information processing tasks for which the quantum discord provides a bound on the performance: we proved that a disturbance-based discord quantifier bounds the CDP. Several questions remain open, specifically whether the CDP is actually equal to the lowest operator-Schmidt-coefficient of the state, and which channels are the hardest to discriminate for a state that has a non-zero CDP. Finally, while our work is strictly related to tomography, the CDP is defined in terms of worst-case channel discrimination. It would be interesting to consider more in general how a probe-ancilla state induces a mapping between a metric on the space of channels and a metric in the space of output probe-ancilla states.

\emph{Acknowledgements.} We thank Vern Paulsen for correspondence and John Watrous for discussions. We acknowledge financial support from the European Union's Horizon 2020 Research
and Innovation Programme under the Marie Sk{\l}odowska-Curie
Action OPERACQC (Grant Agreement No. 661338), and from the Foundational
Questions Institute under the Physics of the Observer
Programme (Grant number FQXi-RFP-1601A).

\appendix


\section*{Appendix}

In this appendix we provide the proofs of the claims made in the main text. It will be convenient to work directly with norms, e.g. $\|X\|_1$, rather than with derived distances, e.g., rather than in terms of the trace distance between two states $\rho$ and $\sigma$ defined as 
\[
D(\rho,\sigma):=\frac{1}{2}\|\rho-\sigma\|_1.
\]
It is useful to recall that, for a Hermitian operators $X=X^\dagger$, one has
\[
\|X\|_1 = \max_{-\mathbb{1} \leq M \leq \mathbb{1}}|\Tr(MX)|.
\]

We define the (Hermitian) super-operator 1-norm of an Hermiticity preserving map $\Gamma$ as~\cite{watrous}
\[
\|\Gamma\|_1 = \sup_{X=X^\dagger; \|X\|_1=1}\|\Gamma[X]\|_1.
\]
Notice that this is equivalent to
\[
\|\Gamma\|_1 = \sup_{X=X^\dagger\neq 0}\frac{\|\Gamma[X]\|_1}{\|X\|_1}.
\]
It is immediate to argue by convexity that the best input $X$ can always be taken to be a pure normalized state $\proj{\psi}$~\cite{watrous}.

We define the diamond norm of an Hermiticity preserving map $\Gamma$ as
\[
\|\Gamma\|_\diamond:=\sup_{n} \|\Gamma\otimes \id_{\mathbb{C}^n}\|_1
\]
where the supremum is over the dimension of the ancillary space $\mathbb{C}^n$. It is easily argued that one can choose $n$ to be equal to the input dimension of the map $\Gamma$~\cite{watrous}. 
Notice that by definition we have the following.
 \begin{Proposition}\label{watt}
 	Let $\Gamma$ be any Hermiticity preserving map, and $X_{AB}$ Hermitian. Then
 	\begin{equation*}
 	\bnorm{\Gamma_A\otimes\text{\emph{id}}_B[X_{AB}]}{1}\leq\bnorm{\Gamma}{\diamond}\bnorm{X_{AB}}{1}.
 	\end{equation*}
 \end{Proposition}
As notation goes, we will indicate the difference of two channels $\Lambda_0$ and $\Lambda_1$ as $\Delta = \Lambda_0 - \Lambda_1$. The channel discrimination power can then be expressed as
\[
\CDP_A (\rho_{AB})= \inf_{\Delta}\frac{\|\Delta_A\otimes\id_B[\rho_{AB}]\|_1}{\|\Delta\|_\diamond}.
\]

\subsection{Continuity of the channel discrimination power}

\begin{Proposition}\label{continuity}
$\CDP_{A}(\rho)$ is continuous:
\begin{equation*}
|\CDP_A(\rab)-\CDP_A(\sab)|\leq\bnorm{\rab-\sab}{1},
\end{equation*}
for any two states $\rho_{AB}$ and $\sigma_{AB}$.
\end{Proposition}
\begin{proof}
Because of the triangle inequality and Proposition~\ref{watt}, one has
\begin{align*}
\bnorm{\doi[\rab]}{1}&=\bnorm{\doi[\sab]+\doi[\rab-\sab]}{1}\\
&\leq\bnorm{\doi[\sab]}{1}+\bnorm{\Delta}{\diamond}\bnorm{\rab-\sab}{1},
\end{align*}
that is 
\begin{align*}
\frac{\bnorm{\doi[\rab]}{1}-\bnorm{\doi[\sab]}{1}}{\|\Delta\|_\diamond}\leq\bnorm{\rab-\sab}{1},
\end{align*}
and the thesis follows.

\end{proof}

\subsection{Monotonicity of the channel discrimination power}

\begin{Proposition}\label{monotonicity}
$\CDP_A(\rho)$ is monotone under local channels on $B$:
\begin{equation}\label{ncdp}
\CDP_A(\id_A\otimes\Lambda_B [\rab])\leq \CDP_A(\rab).
\end{equation}
\end{Proposition}
\begin{proof}
This comes directly from the monotonicity of the trace norm of Hermitian operators under channels, i.e. from $\bnorm{\Lambda[X]}{1}\leq\bnorm{X}{1}$. One has
\begin{multline}
\|\Delta_A\otimes\id_B[\id_A\otimes\Lambda_B[\rab]]\|_1\\
\begin{aligned}
&=\bnorm{\id_A\otimes\Lambda_B[\Delta_A\otimes\id_B[\rab]]}{1}\\
&\leq\bnorm{\Delta_A\otimes\id_B[\rab]}{1},\\
\end{aligned}
\end{multline}
for any $\Delta = \Lambda_0 - \Lambda_1$, and the thesis follows.
\end{proof}

\begin{Proposition}
	The channel discrimination power $\CDP_A$ is invariant under local unitaries on $A$.
\end{Proposition}
\begin{proof}
	For any map $\Lambda$ on $A$ and any unitary $U$ on $A$ we can consider the map $\Lambda'[\cdot]=\Lambda[U^\dagger \cdot U]$ such that $(\Lambda_A\otimes\id_B)[\rho_{AB}]=(\Lambda'_A\otimes\id_B)[U_{A}\rho_{AB}U_A^\dagger]$. Given the freedom in the minimization through which $\CDP_A$ is defined, the claim follows immediately.
\end{proof}

\subsection{Channel discrimination power of pure states}

For pure states the CDP can be computed exactly. We will need the following lemma, which is a slight generalization of observations in, e.g., Ref.~\cite{bph}.

\begin{Lemma}\label{brandao2}
Let $\ket{\psi}_{AA'}=\sum_{k=1}^{d}\sqrt{p_k}\ket{a_k}_{A}\otimes\ket{b_k}_{A'}$ be a pure state with $d=d_A=d_{A'}$, and the Schmidt coefficients ordered as $p_1\geq p_2\geq\ldots\geq p_d$. Then
\begin{equation*}
p_d\bnorm{\Delta}{\diamond}\leq \bnorm{\Delta\otimes\text{\emph{id}}[\proj{\psi}]}{1}.
\end{equation*}
\end{Lemma}
\begin{proof}
We use the fact that any pure state $\ket\psi_{AA'}$ can be expressed as
\begin{align}\label{14}
\ket\psi_{AA'}=(\mathds{1}\otimes C)\ket{\tilde{\psi}^+}_{AA'},
\end{align} 
with $\ket{\tilde{\psi}^+}_{AA'} = \sum_{k=1}^{d}\ket{k}_A\otimes\ket{k}_{A'}$, and $C=\sum_{l=1}^d \sqrt{p_l} \ket{b_l}\bra{a_l^*}$, where $\ket{a_l^*}$ is the basis state whose coefficients in the basis $\ket{k}$ are the complex conjugates of those of $\ket{a_l}$. Notice that the singular values of $C$ coincide with the Schmidt coefficients of $\ket{\psi}$, and the fact that $\ket{\psi}$ is normalized implies $\|C\|_2 = 1$, hence $\|C\|_\infty \leq 1$. 

The claim is trivial if $p_d=0$. If $p_d>0$, then $C$ is invertible, and we can express any other state $\ket{\phi}_{AA'}=(\mathds{1}\otimes D)\ket{\tilde{\psi}^+}_{AA'}$ as
\[
\ket{\phi}_{AA'} = (\mathds{1}\otimes DC^{-1})\ket{\psi}_{AA'} 
\]
Let $\ket{\phi}_{AA'}$ be the state that achieves the diamond norm $\|\Delta\|_\diamond$, that is $\|\Delta\|_\diamond = \|\Delta_A\otimes \id_{A'}[\proj{\phi}_{AA'}] \|_1$. Then
\begin{multline*}
\bnorm{\Delta}{\diamond} \\
\begin{aligned}
& = \|\Delta_A\otimes \id_{A'}[\proj{\phi}_{AA'}] \|_1\\
& = \|(\mathds{1}\otimes DC^{-1})\big(\Delta_A\otimes \id_{A'}[\proj{\psi}_{AA'}]\big)(\mathds{1}\otimes DC^{-1})^\dagger\|_1\\
& \leq \|\openone \otimes DC^{-1} \|_\infty^2\|\Delta_A\otimes \id_{A'}[\proj{\psi}_{AA'}]\|_1\\
& \leq \|D\|_\infty^2 \|C^{-1} \|_\infty^2\|\Delta_A\otimes \id_{A'}[\proj{\psi}_{AA'}]\|_1\\
&= p_d^{-1}\|\Delta_A\otimes \id_{A'}[\proj{\psi}_{AA'}]\|_1
\end{aligned}
\end{multline*}
where in the first inequality we have used the H\"older's inequality, $|\Tr(XY)|\leq \|X\|_\infty \|Y\|_1$, twice. For the last line, just observe that the largest singular value of $C^{-1}$ is the reciprocal of the smallest singular value of $C$.
\end{proof}

\begin{Theorem}\label{cdppure}
Let $\ket{\psi}=\sum_{k}\sqrt{p_k}\ket{a_k}\otimes \ket{b_k}$ be a bipartite state vector in its Schmidt decomposition. Then 
\begin{equation}\label{lbl}
\CDP_A(\proj{\psi})=p_{d_A}.
\end{equation}
\end{Theorem}
\begin{proof}
Lemma \ref{brandao2} implies immediately $\CDP_A(\proj{\psi})\geq p_{d_A}$.
We will prove the inequality in the other direction, that is, $\CDP_A(\proj{\psi})\leq p_{d_A}$, by constructing a pair of perfectly distinguishable channels that are hard to distinguish by means of $\ket{\psi}$. We observe that, because in the case of pure states $\CDP_A$ only depends on the Schmidt coefficients, we can assume $\ket{a_k}=\ket{b_k}=\ket{k}$, without loss of generality. Let us introduce the channels
\begin{equation}
	\begin{aligned}
		\Lam_0[X]&=\Tr[PX]\proj{2}+\Tr[(\mathds{1}-P)X]\proj{0}\\
		\Lam_1[X]&=\Tr[PX]\proj{2}+\Tr[(\mathds{1}-P)X]\proj{1},
	\end{aligned}
	\label{eq:mapsupperbound}
\end{equation}
with $P=\sum_{i=1}^{d_A-1}\proj{i}$ and $\openone - P = \proj{d_A}$. Then, $\Delta[X]=\bra{d_A}X\ket{d_A}(\proj{0}-\proj{1})$. It is clear by their definition that the two channels are perfectly distinguishable, even without the use of an ancilla, since
\[
\Lambda_0[\proj{d_A}]=\proj{0},\quad \Lambda_1[\proj{d_A}]=\proj{1},
\]
so that $\|\Lambda_0-\Lambda_1\|_\diamond = \|\Lambda_0-\Lambda_1\|_1 = 2$.
On the other hand, 
\begin{multline*}
\bnorm{(\Lambda_0-\Lambda_1)\otimes\id\proj{\psi}}{1}\\
\begin{split}
&=\|(\proj{0}-\proj{1}) \otimes \Tr_A(\proj{d_A}_A\proj{\psi}_{AB})\|_1\nonumber\\
&=p_{d_A}\bnorm{(\proj{0}-\proj{1})\otimes \proj{d_A})}{1}\nonumber\\
&=2p_{d_A}.\nonumber
\end{split}
\end{multline*}
Thus, we have proven that it must be $\CDP_A(\proj{\psi})\leq p_{d_A}$. 
\end{proof}

\section{The channel discrimination power is maximal for maximally entangled states}

It is known that every extension $\rab$ of $\ra$ is obtained by a channel  acting on a purification of $\ra$. We provide a proof for completeness.

\begin{Proposition}\label{prolem}
The following are equivalent:
\begin{enumerate}[(i)]
\item There is a channel $\Lambda_{A'\rightarrow B}$ such that $\rab=(\id_A\otimes\Lambda_{A'\rightarrow B})[\Psi_{AA'}]$, for $\Psi_{AA'}$ a purification of $\rho_A$;
\item $\ra=\Tr_B(\rab)$. 
\end{enumerate}
\end{Proposition}
\begin{proof}
That (i) implies (ii) is immediate, because $\Lambda_{A'\rightarrow B}$ is trace preserving.

For the reverse implication, consider a purification $\Phi_{ABC}$ of $\rab$. It is clear that $\Phi_{ABC}$ is also a purification of $\rho_A$.
We know from Uhllman's theorem~\cite{uhlmann1976transition} that different purifications of the same state are connected by a unitary transformation (technically speaking, unless the two spaces considered have the same dimensions, an isometry); hence we can write
\begin{align*}
\rab&=\Tr_C(\sigma_{ABC})\\
&=\Tr_C[(\mathds{1}\otimes U_{A'\rightarrow BC})\Psi_{AA'}(\mathds{1}\otimes U_{A'\rightarrow BC})^\dag]\\
&=(\id_A\otimes\Lambda_{A'\rightarrow B})[\Psi_{AA'}],
\end{align*}
with $\Lambda_{A'\rightarrow B} [\cdot]\deff \Tr_C[U_{A'\rightarrow BC} \cdot U_{A'\rightarrow BC}^\dagger]$.
\end{proof}

\begin{Theorem}
The channel discrimination power $\CDP_A$ is maximal for maximally entangled states, for which it is equal to $1/d_A$.
\end{Theorem}
\begin{proof}
Given Propositions \ref{monotonicity} and \ref{prolem}, it is clear that the maximum of the channel discrimination power is achieved by pure states. On the other hand, Theorem \ref{cdppure} tells us that the CDP of a pure state is equivalent to the (square) of the last Schmidt coefficient. The latter cannot be bigger than $1/d_A$, which is achieved for a maximally entangled state.
\end{proof}

\subsection{Bounds for the channel discrimination power of mixed states}

\begin{Theorem}\label{appendix:16}
Let $\rab=\sum_ir_iA_i\otimes B_i$ be the OSD of $\rab$, with $\{A_i\}$, $\{B_i\}$ Hermitian orthonormal bases for $\cL(\cH_A)$ and $\cL(\cH_B)$, respectively. Then 
\begin{equation}
\frac{r_{d_A^2}}{d^{5/2}_A} \leq \CDP_A(\rab)\leq \min_i\left\{r_i\frac{\bnorm{B_i}{1}}{\bnorm{A_i}{\infty}}\right\} \leq r_{d_A^2}\sqrt{d_Ad_B}.
\end{equation}
\end{Theorem}
\begin{proof}

We first prove $ r_{d_A^2}/d_A^{5/2} \leq \CDP_A(\rab)$.

We start by finding a lower bound for the numerator in the definition of the $\CDP_A(\rab)$. First, observe that
\begin{multline*}
\|\Delta\otimes \id[\rab]\|_1\\
\begin{aligned}
&=\bnorm{\sum_ir_i\Delta(A_i)\otimes B_i}{1}\nonumber \\
& =\max_{-\openone \leq M_{AB} \leq \openone} \left|\Tr\left(M_{AB} \sum_ir_i\Delta(A_i)\otimes B_i \right)\right| \\
& \geq \max_{\substack{-\openone \leq M_A \leq \openone\\ -\openone \leq M_{B} \leq \openone}} \left|\Tr\left(M_A\otimes M_B \sum_i r_i\Delta(A_i)\otimes B_i \right)\right| \\
& \geq \max_i \left\{ r_i \frac{ \| \Delta[A_i] \|_1}{\|B_i\|_\infty} \right\}. \\
& \geq r_{d_A^2} \max_i \|\Delta[A_i]\|_1
\end{aligned}
\end{multline*}
The first inequality is due to restricting the class of operators $M_{AB}$ to be product. The second inequality is due to further choosing $M_A$ such that $\|\Delta[A_k]\|_1 = |\Tr(M_A \Delta[A_k] )|$ and $M_B = B_k / \|B_k\|_\infty$, with $k$ the index such that the maximum over $i$ in the last line is achieved.  Notice that, because of the orthonormality of the $B_i's$, this choice for $M_B$ selects only one term in the sum. The last inequality is due to the fact that $\|B_i\|_\infty \leq \|B_i\|_2 = 1$, and that $r_i \geq r_{d_A^2}$ by assumption.

The maximally entangled state can be expressed as $\proj{\psi^+}=\frac{1}{d_A}\sum_{i=1}^{d_A^2}C_i\otimes{C^*_i}$ for any orthonormal operator basis $\{C_k\}\subset \cL(\cH_A)$, in particular for the one appearing in the OSD of $\rab$. Thus,  using Lemma \ref{brandao2},
\begin{align*}
\bnorm{\Delta}{\diamond}
&\leq d_A\bnorm{\doi[\proj{\psi^+}]}{1}\nonumber\\
&=d_A\bnorm{\frac{1}{d_A}\sum_{i}\Delta\left[A_i\right]\otimes A^*_i}{1}\nonumber\\
& \leq \sum_{i}\|\Delta\left[A_i\right]\|_1\|A_i^*\|_1\nonumber\\
&\leq d^{5/2}_A\max_i \bnorm{\Delta[A_i]}{1},
\end{align*} 
having used the triangle inequality, the fact that there are $d_A^2$ terms in the sum, and that $\|A_i^*\|_1 = \|A_i\|_1 \leq \sqrt{d_A} \|A_i\|_2 = \sqrt{d_A}$.
Thus, combining the above,
\begin{align*}
\CDP_A(\rab)&=\inf_{\Delta}\frac{\bnorm{\doi[\rab]}{1}}{\bnorm{\Delta}{\diamond}}\geq \frac{r_{d_A^2}}{d^{5/2}_A},
\end{align*}
which completes the first part of the theorem.

We now show how to upper bound the CDP. To do that, let us consider the following channels:
$$\Lam_{i}[X]=\Tr(X)\frac{\mathds{1}}{d_A}+\epsilon\Tr(A_l X)Y_{i},$$
for $i=0,1$, with traceless Hermitian operators $Y_0$ and $Y_1$, and $A_l$ is the local basis operator of the OSD of $\rab$ corresponding to the $l$\textsuperscript{th} OSC $r_l$. Such maps are trace-preserving by construction, and completely positive for $\epsilon$ small enough, e.g. for $\epsilon \leq 1/(d_A \|A_l\|_\infty \|\max\{\|Y_0\|_\infty, \|Y_1\|_\infty\})$. Then,
$$
\Delta[X]=\epsilon\Tr(A_{l} X)(Y_0-Y_1),
$$
and 
\begin{align}\label{tre}
\bnorm{\doi[\rab]}{1}&=\epsilon\bnorm{\sum_i r_i (Y_0-Y_1)\Tr(A_{l}A_i)\otimes B_i}{1}\nonumber\\
&=\epsilon\bnorm{r_l(Y_0-Y_1)\otimes B_{l}}{1}\nonumber\\
&=r_l\epsilon\bnorm{Y_0-Y_1}{1}\bnorm{B_{l}}{1}.
\end{align}
On the other hand, we claim that 
\begin{equation}\label{quattro}
\bnorm{\Delta}{\diamond}=\epsilon\bnorm{Y_0-Y_1}{1}\bnorm{A_{l}}{\infty}.
\end{equation}
Before proving such claim, let us notice that Eqs. \eqref{tre} and \eqref{quattro} complete the proof of the theorem. Indeed, by recalling the definition of the CDP and using Eqs. \eqref{tre} and \eqref{quattro}, one  gets 
\begin{align*}
\CDP_A(\rab)&\leq
r_l\frac{\bnorm{B_l}{1}}{\bnorm{A_l}{\infty}},
\end{align*}
for any $l$, that is
\begin{align*}
\CDP_A(\rab)&\leq
\min_i\left\{r_i\frac{\bnorm{B_i}{1}}{\bnorm{A_i}{\infty}}\right\}.
\end{align*}
We observe that the right-hand side can be itself upper bounded:
\begin{align*}
\min_i\left\{r_i\frac{\bnorm{B_i}{1}}{\bnorm{A_i}{\infty}}\right\}
&\leq r_{d_A^2}\frac{\bnorm{B_d}{1}}{\bnorm{A_d}{\infty}}\\
&\leq r_{d_A^2}\frac{d_B^{1/2}\bnorm{B_d}{2}}{d_A^{-1/2}\bnorm{A_d}{2}}\\
&=r_{d_A^2}(d_Ad_B)^{1/2},
\end{align*}
where we have used properties of the $p$-norms in the second inequality.

We now prove \reff{quattro}. To do so, let us consider an arbitrary
\[
\begin{aligned}
\ket{\psi}
&=\sum_i\sqrt{p_i}\ket{a_i}\ket{b_i}\\
&=(\mathbb{1}\otimes C) \ket{\tilde{\psi}^+}
\end{aligned}
\]
where $\|C\|_2=1$ for $\ket{\psi}$ to be normalized (see the proof of Lemma~\ref{brandao2}).
Notice that
\[
\begin{split}
\|\doi[\proj{\psi}]\|_1
& =\| (\mathbb{1}\otimes C) (\doi[\proj{\tilde\psi^+}]) (\mathbb{1}\otimes C)^\dagger\|_1 \\
& = \|\epsilon (Y_0 - Y_1 ) \otimes CA_l^T C^\dagger\|_1 \\
& = \epsilon \|Y_0 - Y_1 \|_1 \|CA_l^T C^\dagger\|_1. \\
\end{split}
\]
Thus, it is sufficient to prove that, for a given $X=X^\dagger$, 
\[
\max_{\|C\|_2=1} \|C X C^\dagger\|_1 = \|X\|_\infty.
\]
Notice that $\|X\|_\infty = \|X^T\|_\infty$.

Let $\ket{x}$ be the eigenvector of $X$ corresponding to the largest eigenvalue (in modulus) $\|X\|_\infty$. Choosing $C=\proj{x}$ we have  $\|C X C^\dagger \|_1= \|\proj{x} X\proj{x}\|_1 = \|X\|_\infty$, thus $\max_{\|C\|_2=1} \|C X C^\dagger\|_1 \geq \|X\|_\infty$. 

To prove the other direction, it is useful to recall the polar decomposition $C = U\sqrt{\rho}$, for $\rho$ a normalized state and $U$ a unitary, and the unitary invariance of the $p$-norms, so that what we aim to prove can be cast as
\[
\begin{aligned}
\max_{\|C\|_2=1} \|C X C^\dagger\|_1
& = \max_{\rho \geq 0, \Tr(\rho)=1} \|\sqrt{\rho} X \sqrt{\rho}\|_1 \\
& \leq \|X\|_\infty.
\end{aligned}
\]

Let us also recall that a Hermitian matrix can be expressed as the difference of two positive semidefinite matrices with orthogonal support:
\begin{align*}
X=X^+-X^-,
\end{align*}
with $X^\pm\geq0$, $X^+X^- = X^-X^+ = 0$.
Then,
\begin{align}
\bnorm{\sqrt{\rho}X\sqrt{\rho}}{1}&=\bnorm{\sqrt{\rho}(X^+-X^-)\sqrt{\rho}}{1}\nonumber\\
&=\bnorm{\sqrt{\rho}X^+\sqrt{\rho}-\sqrt{\rho}X^-\sqrt{\rho}}{1}\nonumber\\
&\leq\bnorm{\sqrt{\rho}X^+\sqrt{\rho}}{1}+\bnorm{\sqrt{\rho}X^-\sqrt{\rho}}{1}\nonumber\\
&=\Tr(\rho X^+)+\Tr(\rho X^-)\nonumber\\
&=\Tr(\rho(X^++X^-))\nonumber\\
&\leq\bnorm{X}{\infty}\Tr(\rho)\nonumber\\
&=\bnorm{X}{\infty}.\label{sette}
\end{align}

\noindent
In the second to last line we have used $0\leq X^++X^-\leq\bnorm{X}{\infty}\cdot\mathds{1}$. 
We have proved the claim in \reff{quattro}, hence the theorem.

\end{proof}

\subsection{Bound on the channel discrimination power based on disturbance by operations that reduce the operator Schmidt rank}

Here we want to study the behaviour of the CDP under the action of maps that reduce the OSR. 
\begin{Theorem} 
We have
\[
\CDP_A(\rab)\leq  \min_{\substack{\Lam \text{\ s.t.} \\ \OSR\left(\Lam\otimes \id[\rab]\right) < d_A^2}}\bnorm{\rab-\Lam\otimes\id[\rab]}{1}
\]
where the minimization is over all channels that acting on $A$ reduce the OSR of $\rho_{AB}$ to less than maximal.
\end{Theorem}
\begin{proof}
It holds
\begin{align*}
\bnorm{\doi[\rab]}{1}
&\leq\bnorm{\doi[\rab-\Lam\otimes\id[\rab]]}{1} \\
& \quad + \bnorm{(\Delta\circ\Lam)\otimes\id[\rab]}{1}\nonumber\\
&\leq \bnorm{\Delta}{\diamond}\bnorm{\rab-\Lam\otimes\id[\rab]}{1} \\
& \quad + \bnorm{\Delta\otimes\id[\Lambda\otimes\id[\rab]]}{1},\label{otto}
\end{align*}
having used Proposition \ref{watt}. Then,
\begin{align*}
\inf_{\Delta}\frac{\bnorm{\Delta\otimes\id[\rab]}{1}}{\bnorm{\Delta}{\diamond}}
&\leq\bnorm{\rab-\Lam\otimes\id[\rab]}{1} \\
& \quad +\inf_\Delta \frac{\bnorm{\Delta\otimes\id[\Lambda\otimes\id[\rab]]}{1}}{\bnorm{\Delta}{\diamond}}\nonumber\\
&=\bnorm{\rab-\Lam\otimes\id[\rab]}{1}, \\
\end{align*}
\noindent
where we have used that the CDP of $\Lambda\otimes\id[\rab]$ (the second term on the right-hand side of the inequality) vanishes under the assumption $ \OSR\left(\Lam\otimes \id\,[\rab]\right) < d_A^2$.
It finally follows
\begin{equation*}
\CDP_A(\rab)\leq \min_{\substack{\Lam \text{\ s.t.} \\ \OSR\left(\Lam\otimes \emph{id}\,[\rab]\right)\,<\, d_A^2}}\bnorm{\rab-\Lam\otimes\id[\rab]}{1}.
\end{equation*}
\end{proof}
As a particular example involving the last theorem, let $\Pi[X]=\sum_{i=1}^d\proj{i}X\proj{i}$ be the channel which dephases in an arbitrary basis. Then 
\begin{equation*}
\CDP_A(\rab)\hspace{.2cm}\leq  \min_{\Pi_A\otimes\id_B}\bnorm{\rab-\Pi_A\otimes\id_B[\rab]}{1}.
\end{equation*}

\subsection{CDP bounds for isotropic states}

We are considering isotropic states, i.e. states of the form
\begin{equation}\label{isos}
\rho_{\text{iso}}(p)=\frac{1-p}{d^2}\,\mathds{1}+p \proj{\psi^+},
\end{equation}
where $0\leq p\leq 1$ and $\ket{\psi^+}$ is the standard maximally entangled state.
It is known and immediate to check that
\[
\proj{\psi^+} = \frac{1}{d}\sum_{k=1}^{d^2} A_k\otimes A^*_k
\]
for any orthonormal operator basis $\{A_k\}$, with complex conjugation taken in the local Schmidt basis of the maximally entangled state. We can choose $A_1 = \frac{\openone}{\sqrt{d}}$, and find immediately
\begin{equation}
\label{eq:OSDiso}
\rho_{\text{iso}}(p)=\frac{1}{d} \frac{\openone}{\sqrt{d}} \otimes\frac{\openone}{\sqrt{d}}
+
\frac{p}{d}\sum_{k=2}^{d^2} A_k\otimes A^*_k,
\end{equation}
where $\{A_k\}$ is any collection of $d^2-1$ traceless orthonormal operators. Thus, the OSCs of $\rho_{\text{iso}}(p)$ are evidently $(1/d,p/d,\ldots,p/d)$.

\begin{Theorem}
For the isotropic state it holds
\begin{equation}
\frac{p}{d+1-p}
\leq
\CDP_A(\rho_{\text{iso}}(p))
\leq\min\left\{2\frac{p}{d},\frac{1}{d}\right\}.
\end{equation}
\end{Theorem}

\begin{proof}
We start by proving the upper bound. That $\CDP_A(\rho_{\text{iso}}(p))\leq 1/d$ can be straightforwardly be verified by using the same two maps \eqref{eq:mapsupperbound} that were used to prove the upper bound for pure states. In order to prove $\CDP_A(\rho_{\text{iso}}(p)) \leq 2p/d$, we will use the bound $\CDP_A(\rab) \leq \min_i\left\{r_i\frac{\bnorm{B_i}{1}}{\bnorm{A_i}{\infty}}\right\}$ from Theorem \ref{appendix:16}, exploiting the freedom in choosing the decomposition \eqref{eq:OSDiso}. E.g., we can choose $A_2 = (\ket{1}\bra{2}+\ket{2}\bra{1})/\sqrt{2}$, with $B_2 = A_2^* = A_2$, so that $\bnorm{A_2}{\infty}=1/\sqrt{2}$ and $\bnorm{B_2}{1} = \sqrt{2}$. Thus,
\[
\CDP_A(\rho_{\text{iso}}(p)) \leq r_2 \frac{\bnorm{B_2}{1}}{\bnorm{A_2}{\infty}} = \frac{p}{d} \frac{\bnorm{A_2}{1}}{\bnorm{A_2}{\infty}} = \frac{p}{d} 2.
\]

For the lower bound, we generalize the approach of Lemma~\ref{brandao2}.

Given two arbitrary channels, let $\proj{\psi}$ be optimal for the diamond norm of their difference, i.e.
\begin{equation*}
\bnorm{\Delta}{\diamond}=\sup_\rho\bnorm{\doi [\rho] }{1}=\bnorm{\doi [\proj{\psi}] }{1}
\end{equation*}
and let us consider $C$ such that
\[
\ket\psi_{AA'}=(\mathds{1}\otimes C)\ket{\tilde{\psi}^+}_{AA'}
\]
Notice that $\Tr_A(\proj{\psi})=CC^\dagger$, with $CC^\dagger \geq 0 $ a normalized state.

Let us define the state
\begin{align*}
\sigma(p)&\deff(1-p)\frac{\mathds{1}}{d}\otimes CC^\dagger+p\proj{\psi}\\
&=d(\mathds{1}\otimes{C})\left[(1-p)\frac{\mathds{1}}{d}\otimes\frac{\mathds{1}}{d}+p\mes\right](\mathds{1}\otimes{C^\dagger})\\
&=d(\mathds{1}\otimes{C})\ \rho_{\text{iso}}(p)\ (\mathds{1}\otimes{C^\dagger}).
\end{align*}
Then,
\begin{align*}
\proj{\psi}
&=\frac{1}{p}\left[\sigma{(p)}-(1-p)\frac{\mathds{1}}{d}\otimes CC^\dagger\right],
\end{align*}
and
\begin{align*}
\bnorm{\Delta}{\diamond}
&=\bnorm{\doi [\proj{\psi}] }{1}\\
&=\bnorm{\frac{1}{p}\left[\doi[\sigma{(p)}]-(1-p)\Delta\left[\frac{\mathds{1}}{d}\right]\otimes CC^\dagger\right]}{1}\\
&\leq \frac{1}{p}\bnorm{\doi[\sigma{(p)}]}{1}+\frac{1-p}{p}\bnorm{\Delta\left[\frac{\mathds{1}}{d}\right] }{1} \\
&=\frac{d}{p}\bnorm{(\mathds{1}\otimes{C})\ \doi[\rho_{\text{iso}}(p)]
\ (\mathds{1}\otimes{C^\dagger})}{1}\\
&\quad+\frac{1-p}{p}\bnorm{\Delta\left[\frac{\mathds{1}}{d}\right]}{1}\\
&\leq\frac{d}{p}\bnorm{C}{\infty}^2
\bnorm{\doi[\rho_{\text{iso}}(p)]}{1}+\frac{1-p}{p}\bnorm{\Delta\left[\frac{\mathds{1}}{d}\right]}{1}\\
&\leq \frac{d}{p}\bnorm{\doi[\rho_{\text{iso}}(p)}{1}+\frac{1-p}{p}\bnorm{\Delta\left[\frac{\mathds{1}}{d}\right]}{1}.
\end{align*}
Finally, since $\frac{\mathds{1}}{d}=\Tr_B(\rho_{\text{iso}}(p))$ and the partial trace is a channel, the monotonicity of the trace distance implies
\begin{align*}
\bnorm{\Delta\left[\frac{\mathds{1}}{d}\right]}{1}&=\bnorm{\Delta_A\left[\Tr_B(\rho_{\text{iso}}(p))\right]}{1}\nonumber\\
&=\bnorm{\Tr_B\left(\Delta_A[\rho_{\text{iso}}(p)]\right)}{1}\nonumber\\
&\leq\bnorm{\doi[\rho_{\text{iso}}(p)]}{1},
\end{align*}
Thus,
\begin{equation*}
\bnorm{\Delta}{\diamond}\leq \left(\frac{d+1-p}{p}\right)\bnorm{\doi[\rho_{\text{iso}}(p)]}{1},
\end{equation*}
from which we obtain
\begin{equation*}
\CDP_A(\rho_{\text{iso}}(p))\geq\frac{p}{d+1-p}.
\end{equation*}
\end{proof}

\section{Bound on the channel discrimination power for states that satisfy the realignment criterion of separability}

In this section we provide the tools to prove the bound on the channel discrimination power of states that respect the realignment criterion of separability. These include, obviously, all separable states, but also ``weakly'' entangled states that are not detected by the realignment criterion.

\begin{Lemma}\label{lem:dist}
Let $r_i(\rho_{AB})$ be the ordered operator Schmidt coefficients of $\rho$. Then
\[
\sum_{i\geq2}r^2_i(\rho_{AB}) = \Tr(\rho^2)- r_1^2 \leq \|\rho_{AB} - \sigma_A\otimes \sigma_B\|^2_2,
\]
for any product state $\sigma_A\otimes \sigma_B$.
\end{Lemma}
\begin{proof}
We recall that the OSCs $r_i(\rho_{AB})$ are the singular values of the correlation matrix $[C_{ij}(\rho_{AB})]_{ij}$, with
\[
C_{ij}(\rho_{AB}):=\llangle F_i\otimes G_j | \rho_{AB}\rrangle,
\]
where $\{F_i\}$ and $\{G_j\}$ are arbitrary local orthonormal bases for operators. We will use that, for any two matrices $M$ and $N$, with ordered singular values $\sigma_i(M)$ and $\sigma_i(N)$, respectively, it holds (see Corollary 7.3.5 in~\cite{horn2012matrix}), 
\[
\sum_i (\sigma_i(M) - \sigma_i(N))^2 \leq \|M-N\|^2_2.
\]
Notice that $r_i(\sigma_A\otimes\sigma_B)= 0$, for $i\geq 2$.
Thus,
\[
\begin{split}
\sum_{i\geq 2}r^2_i(\rho_{AB})
&= \sum_{i\geq 2}(r_i(\rho_{AB}) - r_i(\sigma_A\otimes\sigma_B))^2\\
&\leq \sum_{i}(r_i(\rho_{AB}) - r_i(\sigma_A\otimes\sigma_B))^2 \\
&\leq \| C(\rho_{AB}) - C(\sigma_A\otimes\sigma_B)\|^2_2 \\
&= \| C(\rho_{AB}-\sigma_A\otimes\sigma_B)\|^2_2 \\
&= \| \rho_{AB}-\sigma_A\otimes\sigma_B \|^2_2 ,
\end{split}
\]
having used that $\|C(X)\|_2 = \|X\|_2$ for any $X$.
\end{proof}

\begin{Proposition}\label{pro:boundrd2}
For any state $\rho_{AB}$ on $\mathbb{C}^d\otimes \mathbb{C}^d$, the lowest operator Schmidt coefficient obeys 
\[
r_{d^2} \leq \sqrt{\Tr(\rho^2) - \frac{1}{d^2}}.
\]
\end{Proposition}
\begin{proof}
Immediate, by using Lemma~\ref{lem:dist} in the case $\sigma_A\otimes \sigma_B = \frac{\I}{d}\otimes \frac{\I}{d}$, and the fact that
\[
\begin{split}
\bnorm{\rho_{AB} - \frac{\I}{d}\otimes \frac{\I}{d}}{2}^2
&= \Tr\left(\left(\rho_{AB} - \frac{\I}{d}\otimes \frac{\I}{d}\right)^2\right). \\
&= \Tr(\rho^2) - \frac{1}{d^2}
\end{split}
\]
\end{proof}

\begin{Theorem}
	If the OSCs of $\rab$ satisfy $\sum_ir_i\leq 1$, then $r_{d^2}\leq r_{\textup{CN}}$ with
	\[
	r_{\textup{CN}} = \frac{d(d^2-1)-\sqrt{d^2-1}}{d(d^2-1)^2+d^3}<\frac{1}{d^2}.
	\]
\end{Theorem}
\begin{proof}
We want to find the maximal value $r_{d^2}$ can assume under the condition
\begin{equation}\label{eq:cond1}
\sum_ir_i\leq 1.
\end{equation}
We notice that Proposition~\ref{pro:boundrd2} implies that the OSCs of every state respect
\begin{equation}\label{eq:cond2}
r_{d^2}^2 \leq \sum_i r_i^2 - \frac{1}{d^2}
\end{equation}
(recall that $\Tr(\rho^2) = \sum_i r_i^2$). Thus, we want to find the maximum of $r_{d^2}$ under conditions \eqref{eq:cond1} and \eqref{eq:cond2}. Notice that, by definition, $r_i \geq 0$, and $r_1 \geq r_2 \geq \ldots \geq r_d^2$.

It it clear that the maximum $r_{d^2}$ will be found for the condition \eqref{eq:cond1} being satisfied with equality, since, if the left-hand side of \eqref{eq:cond1} was smaller than 1, then we could increase all the OSCs, including $r_{d^2}$, to make it equal to 1. Moreover, for fixed $r_{d^2}$, the largest value of $\sum_i r_i^2$ is achieved for $r_2 =r_3 = \ldots = r_{d^2} = r $ and $r_1 = 1 - r$. This is due to the fact that $\sum_i r_i^2$  is Schur convex. Thus, we can find the maximal $r_{d^2}$ compatible with the constraints, by finding the largest $r$ such that
\[
r^2 \leq (d^2-1)r^2 +(1 - (d^2-1)r)^2 - \frac{1}{d^2}.
\]
One finds that such a value is given by 
\[
	r_{CN} = \frac{d(d^2-1)-\sqrt{d^2-1}}{d(d^2-1)^2+d^3}<\frac{1}{d^2}.
\]
\end{proof}


\begin{thebibliography}{43}
\expandafter\ifx\csname natexlab\endcsname\relax\def\natexlab#1{#1}\fi
\expandafter\ifx\csname bibnamefont\endcsname\relax
  \def\bibnamefont#1{#1}\fi
\expandafter\ifx\csname bibfnamefont\endcsname\relax
  \def\bibfnamefont#1{#1}\fi
\expandafter\ifx\csname citenamefont\endcsname\relax
  \def\citenamefont#1{#1}\fi
\expandafter\ifx\csname url\endcsname\relax
  \def\url#1{\texttt{#1}}\fi
\expandafter\ifx\csname urlprefix\endcsname\relax\def\urlprefix{URL }\fi
\providecommand{\bibinfo}[2]{#2}
\providecommand{\eprint}[2][]{\url{#2}}

\bibitem[{\citenamefont{Nielsen and Chuang}(2010)}]{nielsen2010quantum}
\bibinfo{author}{\bibfnamefont{M.~A.} \bibnamefont{Nielsen}} \bibnamefont{and}
  \bibinfo{author}{\bibfnamefont{I.~L.} \bibnamefont{Chuang}},
  \emph{\bibinfo{title}{Quantum Computation and Quantum Information}}
  (\bibinfo{publisher}{Cambridge University Press}, \bibinfo{year}{2010}).

\bibitem[{\citenamefont{Giovannetti et~al.}(2004)\citenamefont{Giovannetti,
  Lloyd, and Maccone}}]{giovannetti2004quantum}
\bibinfo{author}{\bibfnamefont{V.}~\bibnamefont{Giovannetti}},
  \bibinfo{author}{\bibfnamefont{S.}~\bibnamefont{Lloyd}}, \bibnamefont{and}
  \bibinfo{author}{\bibfnamefont{L.}~\bibnamefont{Maccone}},
  \bibinfo{journal}{Science} \textbf{\bibinfo{volume}{306}},
  \bibinfo{pages}{1330} (\bibinfo{year}{2004}).

\bibitem[{\citenamefont{T{\'o}th and Apellaniz}(2014)}]{toth2014quantum}
\bibinfo{author}{\bibfnamefont{G.}~\bibnamefont{T{\'o}th}} \bibnamefont{and}
  \bibinfo{author}{\bibfnamefont{I.}~\bibnamefont{Apellaniz}},
  \bibinfo{journal}{Journal of Physics A: Mathematical and Theoretical}
  \textbf{\bibinfo{volume}{47}}, \bibinfo{pages}{424006}
  (\bibinfo{year}{2014}).

\bibitem[{\citenamefont{D'Ariano
  et~al.}(2001{\natexlab{a}})\citenamefont{D'Ariano, LoPresti, and
  Paris}}]{DArianoPP01}
\bibinfo{author}{\bibfnamefont{G.~M.} \bibnamefont{D'Ariano}},
  \bibinfo{author}{\bibfnamefont{P.}~\bibnamefont{LoPresti}}, \bibnamefont{and}
  \bibinfo{author}{\bibfnamefont{M.~G.~A.} \bibnamefont{Paris}},
  \bibinfo{journal}{Phys. Rev. Lett.} \textbf{\bibinfo{volume}{87}},
  \bibinfo{pages}{270404} (\bibinfo{year}{2001}{\natexlab{a}}).

\bibitem[{\citenamefont{Kitaev}(1997)}]{kitaev1997quantum}
\bibinfo{author}{\bibfnamefont{A.~Y.} \bibnamefont{Kitaev}},
  \bibinfo{journal}{Russian Mathematical Surveys}
  \textbf{\bibinfo{volume}{52}}, \bibinfo{pages}{1191} (\bibinfo{year}{1997}).

\bibitem[{\citenamefont{Childs et~al.}(2000)\citenamefont{Childs, Preskill, and
  Renes}}]{childs2000quantum}
\bibinfo{author}{\bibfnamefont{A.~M.} \bibnamefont{Childs}},
  \bibinfo{author}{\bibfnamefont{J.}~\bibnamefont{Preskill}}, \bibnamefont{and}
  \bibinfo{author}{\bibfnamefont{J.}~\bibnamefont{Renes}},
  \bibinfo{journal}{Journal of modern optics} \textbf{\bibinfo{volume}{47}},
  \bibinfo{pages}{155} (\bibinfo{year}{2000}).

\bibitem[{\citenamefont{D'Ariano
  et~al.}(2001{\natexlab{b}})\citenamefont{D'Ariano, Presti, and
  Paris}}]{d2001using}
\bibinfo{author}{\bibfnamefont{G.~M.} \bibnamefont{D'Ariano}},
  \bibinfo{author}{\bibfnamefont{P.~L.} \bibnamefont{Presti}},
  \bibnamefont{and} \bibinfo{author}{\bibfnamefont{M.~G.} \bibnamefont{Paris}},
  \bibinfo{journal}{Physical review letters} \textbf{\bibinfo{volume}{87}},
  \bibinfo{pages}{270404} (\bibinfo{year}{2001}{\natexlab{b}}).

\bibitem[{\citenamefont{Ac{\'\i}n}(2001)}]{acin2001statistical}
\bibinfo{author}{\bibfnamefont{A.}~\bibnamefont{Ac{\'\i}n}},
  \bibinfo{journal}{Physical review letters} \textbf{\bibinfo{volume}{87}},
  \bibinfo{pages}{177901} (\bibinfo{year}{2001}).

\bibitem[{\citenamefont{Gilchrist et~al.}(2005)\citenamefont{Gilchrist,
  Langford, and Nielsen}}]{gilchrist2005distance}
\bibinfo{author}{\bibfnamefont{A.}~\bibnamefont{Gilchrist}},
  \bibinfo{author}{\bibfnamefont{N.~K.} \bibnamefont{Langford}},
  \bibnamefont{and} \bibinfo{author}{\bibfnamefont{M.~A.}
  \bibnamefont{Nielsen}}, \bibinfo{journal}{Physical Review A}
  \textbf{\bibinfo{volume}{71}}, \bibinfo{pages}{062310}
  (\bibinfo{year}{2005}).

\bibitem[{\citenamefont{Rosgen and Watrous}(2005)}]{rosgen2005hardness}
\bibinfo{author}{\bibfnamefont{B.}~\bibnamefont{Rosgen}} \bibnamefont{and}
  \bibinfo{author}{\bibfnamefont{J.}~\bibnamefont{Watrous}}, in
  \emph{\bibinfo{booktitle}{Computational Complexity, 2005. Proceedings.
  Twentieth Annual IEEE Conference on}} (\bibinfo{organization}{IEEE},
  \bibinfo{year}{2005}), pp. \bibinfo{pages}{344--354}.

\bibitem[{\citenamefont{Sacchi}(2005{\natexlab{a}})}]{sacchi2005optimal}
\bibinfo{author}{\bibfnamefont{M.~F.} \bibnamefont{Sacchi}},
  \bibinfo{journal}{Physical Review A} \textbf{\bibinfo{volume}{71}},
  \bibinfo{pages}{062340} (\bibinfo{year}{2005}{\natexlab{a}}).

\bibitem[{\citenamefont{Sacchi}(2005{\natexlab{b}})}]{sacchi2005entanglement}
\bibinfo{author}{\bibfnamefont{M.~F.} \bibnamefont{Sacchi}},
  \bibinfo{journal}{Physical Review A} \textbf{\bibinfo{volume}{72}},
  \bibinfo{pages}{014305} (\bibinfo{year}{2005}{\natexlab{b}}).

\bibitem[{\citenamefont{Lloyd}(2008)}]{lloyd2008enhanced}
\bibinfo{author}{\bibfnamefont{S.}~\bibnamefont{Lloyd}},
  \bibinfo{journal}{Science} \textbf{\bibinfo{volume}{321}},
  \bibinfo{pages}{1463} (\bibinfo{year}{2008}).

\bibitem[{\citenamefont{Rosgen}(2008)}]{rosgen2008additivity}
\bibinfo{author}{\bibfnamefont{B.}~\bibnamefont{Rosgen}},
  \bibinfo{journal}{Journal of Mathematical Physics}
  \textbf{\bibinfo{volume}{49}}, \bibinfo{pages}{102107}
  (\bibinfo{year}{2008}).

\bibitem[{\citenamefont{Watrous}(2005)}]{watrous}
\bibinfo{author}{\bibfnamefont{J.}~\bibnamefont{Watrous}},
  \bibinfo{journal}{Quantum Info. Comput.} \textbf{\bibinfo{volume}{5}},
  \bibinfo{pages}{58} (\bibinfo{year}{2005}), ISSN \bibinfo{issn}{1533-7146},
  \urlprefix\url{http://dl.acm.org/citation.cfm?id=2011608.2011614}.

\bibitem[{\citenamefont{Watrous}(2008)}]{watrous2008distinguishing}
\bibinfo{author}{\bibfnamefont{J.}~\bibnamefont{Watrous}},
  \bibinfo{journal}{Quantum Information \& Computation}
  \textbf{\bibinfo{volume}{8}}, \bibinfo{pages}{819} (\bibinfo{year}{2008}).

\bibitem[{\citenamefont{Choi}(1975)}]{choi1975completely}
\bibinfo{author}{\bibfnamefont{M.-D.} \bibnamefont{Choi}},
  \bibinfo{journal}{Linear algebra and its applications}
  \textbf{\bibinfo{volume}{10}}, \bibinfo{pages}{285} (\bibinfo{year}{1975}).

\bibitem[{\citenamefont{Jamio{\l}kowski}(1972)}]{jamiolkowski1972linear}
\bibinfo{author}{\bibfnamefont{A.}~\bibnamefont{Jamio{\l}kowski}},
  \bibinfo{journal}{Reports on Mathematical Physics}
  \textbf{\bibinfo{volume}{3}}, \bibinfo{pages}{275} (\bibinfo{year}{1972}).

\bibitem[{\citenamefont{Altepeter et~al.}(2003)\citenamefont{Altepeter,
  Branning, Jeffrey, Wei, Kwiat, Thew, O'Brien, Nielsen, and
  White}}]{altepeter2003ancilla}
\bibinfo{author}{\bibfnamefont{J.~B.} \bibnamefont{Altepeter}},
  \bibinfo{author}{\bibfnamefont{D.}~\bibnamefont{Branning}},
  \bibinfo{author}{\bibfnamefont{E.}~\bibnamefont{Jeffrey}},
  \bibinfo{author}{\bibfnamefont{T.}~\bibnamefont{Wei}},
  \bibinfo{author}{\bibfnamefont{P.~G.} \bibnamefont{Kwiat}},
  \bibinfo{author}{\bibfnamefont{R.~T.} \bibnamefont{Thew}},
  \bibinfo{author}{\bibfnamefont{J.~L.} \bibnamefont{O'Brien}},
  \bibinfo{author}{\bibfnamefont{M.~A.} \bibnamefont{Nielsen}},
  \bibnamefont{and} \bibinfo{author}{\bibfnamefont{A.~G.} \bibnamefont{White}},
  \bibinfo{journal}{Physical Review Letters} \textbf{\bibinfo{volume}{90}},
  \bibinfo{pages}{193601} (\bibinfo{year}{2003}).

\bibitem[{\citenamefont{Jen{\v{c}}ov{\'a} and
  Pl{\'a}vala}(2016)}]{jenvcova2016conditions}
\bibinfo{author}{\bibfnamefont{A.}~\bibnamefont{Jen{\v{c}}ov{\'a}}}
  \bibnamefont{and}
  \bibinfo{author}{\bibfnamefont{M.}~\bibnamefont{Pl{\'a}vala}},
  \bibinfo{journal}{Journal of Mathematical Physics}
  \textbf{\bibinfo{volume}{57}}, \bibinfo{pages}{122203}
  (\bibinfo{year}{2016}), \eprint{http://dx.doi.org/10.1063/1.4972286},
  \urlprefix\url{http://dx.doi.org/10.1063/1.4972286}.

\bibitem[{\citenamefont{Chen and Wu}(2003)}]{chen2003matrix}
\bibinfo{author}{\bibfnamefont{K.}~\bibnamefont{Chen}} \bibnamefont{and}
  \bibinfo{author}{\bibfnamefont{L.-A.} \bibnamefont{Wu}},
  \bibinfo{journal}{Quantum Inf. Comput} \textbf{\bibinfo{volume}{3}},
  \bibinfo{pages}{193} (\bibinfo{year}{2003}).

\bibitem[{\citenamefont{Rudolph}(2004)}]{rudolph2004computable}
\bibinfo{author}{\bibfnamefont{O.}~\bibnamefont{Rudolph}},
  \bibinfo{journal}{Letters in Mathematical Physics}
  \textbf{\bibinfo{volume}{70}}, \bibinfo{pages}{57} (\bibinfo{year}{2004}).

\bibitem[{\citenamefont{Modi et~al.}(2012)\citenamefont{Modi, Brodutch, Cable,
  Paterek, and Vedral}}]{RevModPhys.84.1655}
\bibinfo{author}{\bibfnamefont{K.}~\bibnamefont{Modi}},
  \bibinfo{author}{\bibfnamefont{A.}~\bibnamefont{Brodutch}},
  \bibinfo{author}{\bibfnamefont{H.}~\bibnamefont{Cable}},
  \bibinfo{author}{\bibfnamefont{T.}~\bibnamefont{Paterek}}, \bibnamefont{and}
  \bibinfo{author}{\bibfnamefont{V.}~\bibnamefont{Vedral}},
  \bibinfo{journal}{Rev. Mod. Phys.} \textbf{\bibinfo{volume}{84}},
  \bibinfo{pages}{1655} (\bibinfo{year}{2012}),
  \urlprefix\url{http://link.aps.org/doi/10.1103/RevModPhys.84.1655}.

\bibitem[{\citenamefont{Horn and Johnson}(2013)}]{horn2012matrix}
\bibinfo{author}{\bibfnamefont{R.~A.} \bibnamefont{Horn}} \bibnamefont{and}
  \bibinfo{author}{\bibfnamefont{C.~R.} \bibnamefont{Johnson}},
  \emph{\bibinfo{title}{Matrix analysis}} (\bibinfo{publisher}{Cambridge
  University Press}, \bibinfo{year}{2013}).

\bibitem[{\citenamefont{Horodecki et~al.}(2009)\citenamefont{Horodecki,
  Horodecki, Horodecki, and Horodecki}}]{revent}
\bibinfo{author}{\bibfnamefont{R.}~\bibnamefont{Horodecki}},
  \bibinfo{author}{\bibfnamefont{P.}~\bibnamefont{Horodecki}},
  \bibinfo{author}{\bibfnamefont{M.}~\bibnamefont{Horodecki}},
  \bibnamefont{and}
  \bibinfo{author}{\bibfnamefont{K.}~\bibnamefont{Horodecki}},
  \bibinfo{journal}{Rev. Mod. Phys.} \textbf{\bibinfo{volume}{81}},
  \bibinfo{pages}{865} (\bibinfo{year}{2009}),
  \urlprefix\url{http://link.aps.org/doi/10.1103/RevModPhys.81.865}.

\bibitem[{\citenamefont{Aniello and Lupo}(2009)}]{aniello2009relation}
\bibinfo{author}{\bibfnamefont{P.}~\bibnamefont{Aniello}} \bibnamefont{and}
  \bibinfo{author}{\bibfnamefont{C.}~\bibnamefont{Lupo}},
  \bibinfo{journal}{Open Systems \& Information Dynamics}
  \textbf{\bibinfo{volume}{16}}, \bibinfo{pages}{127} (\bibinfo{year}{2009}).

\bibitem[{\citenamefont{Lupo et~al.}(2008)\citenamefont{Lupo, Aniello, and
  Scardicchio}}]{lupo2008bipartite}
\bibinfo{author}{\bibfnamefont{C.}~\bibnamefont{Lupo}},
  \bibinfo{author}{\bibfnamefont{P.}~\bibnamefont{Aniello}}, \bibnamefont{and}
  \bibinfo{author}{\bibfnamefont{A.}~\bibnamefont{Scardicchio}},
  \bibinfo{journal}{Journal of Physics A: Mathematical and Theoretical}
  \textbf{\bibinfo{volume}{41}}, \bibinfo{pages}{415301}
  (\bibinfo{year}{2008}).

\bibitem[{app()}]{appendix}
\bibinfo{howpublished}{See Appendix for proofs}.

\bibitem[{\citenamefont{Daki{\'c} et~al.}(2010)\citenamefont{Daki{\'c}, Vedral,
  and Brukner}}]{dakic2010necessary}
\bibinfo{author}{\bibfnamefont{B.}~\bibnamefont{Daki{\'c}}},
  \bibinfo{author}{\bibfnamefont{V.}~\bibnamefont{Vedral}}, \bibnamefont{and}
  \bibinfo{author}{\bibfnamefont{{\v{C}}.}~\bibnamefont{Brukner}},
  \bibinfo{journal}{Physical review letters} \textbf{\bibinfo{volume}{105}},
  \bibinfo{pages}{190502} (\bibinfo{year}{2010}).

\bibitem[{\citenamefont{Henderson and Vedral}(2001)}]{henderson2001classical}
\bibinfo{author}{\bibfnamefont{L.}~\bibnamefont{Henderson}} \bibnamefont{and}
  \bibinfo{author}{\bibfnamefont{V.}~\bibnamefont{Vedral}},
  \bibinfo{journal}{J. Phys. A: Math. Gen.} \textbf{\bibinfo{volume}{34}},
  \bibinfo{pages}{6899} (\bibinfo{year}{2001}).

\bibitem[{\citenamefont{Ollivier and Zurek}(2001)}]{ollivier2001quantum}
\bibinfo{author}{\bibfnamefont{H.}~\bibnamefont{Ollivier}} \bibnamefont{and}
  \bibinfo{author}{\bibfnamefont{W.~H.} \bibnamefont{Zurek}},
  \bibinfo{journal}{Phys. Rev. Lett.} \textbf{\bibinfo{volume}{88}},
  \bibinfo{pages}{017901} (\bibinfo{year}{2001}).

\bibitem[{\citenamefont{Lanyon et~al.}(2013)\citenamefont{Lanyon, Jurcevic,
  Hempel, Gessner, Vedral, Blatt, and Roos}}]{lanyon2013experimental}
\bibinfo{author}{\bibfnamefont{B.}~\bibnamefont{Lanyon}},
  \bibinfo{author}{\bibfnamefont{P.}~\bibnamefont{Jurcevic}},
  \bibinfo{author}{\bibfnamefont{C.}~\bibnamefont{Hempel}},
  \bibinfo{author}{\bibfnamefont{M.}~\bibnamefont{Gessner}},
  \bibinfo{author}{\bibfnamefont{V.}~\bibnamefont{Vedral}},
  \bibinfo{author}{\bibfnamefont{R.}~\bibnamefont{Blatt}}, \bibnamefont{and}
  \bibinfo{author}{\bibfnamefont{C.}~\bibnamefont{Roos}},
  \bibinfo{journal}{Physical review letters} \textbf{\bibinfo{volume}{111}},
  \bibinfo{pages}{100504} (\bibinfo{year}{2013}).

\bibitem[{\citenamefont{Piani et~al.}(2008)\citenamefont{Piani, Horodecki, and
  Horodecki}}]{piani2008no}
\bibinfo{author}{\bibfnamefont{M.}~\bibnamefont{Piani}},
  \bibinfo{author}{\bibfnamefont{P.}~\bibnamefont{Horodecki}},
  \bibnamefont{and}
  \bibinfo{author}{\bibfnamefont{R.}~\bibnamefont{Horodecki}},
  \bibinfo{journal}{Phys. Rev. Lett.} \textbf{\bibinfo{volume}{100}},
  \bibinfo{pages}{090502} (\bibinfo{year}{2008}).

\bibitem[{\citenamefont{Piani et~al.}(2014)\citenamefont{Piani, Narasimhachar,
  and Calsamiglia}}]{piani2014quantumness}
\bibinfo{author}{\bibfnamefont{M.}~\bibnamefont{Piani}},
  \bibinfo{author}{\bibfnamefont{V.}~\bibnamefont{Narasimhachar}},
  \bibnamefont{and}
  \bibinfo{author}{\bibfnamefont{J.}~\bibnamefont{Calsamiglia}},
  \bibinfo{journal}{New J. Phys.} \textbf{\bibinfo{volume}{16}},
  \bibinfo{pages}{113001} (\bibinfo{year}{2014}),
  \urlprefix\url{http://stacks.iop.org/1367-2630/16/i=11/a=113001}.

\bibitem[{\citenamefont{Boixo et~al.}(2011)\citenamefont{Boixo, Aolita,
  Cavalcanti, Modi, and Winter}}]{boixo2011quantum}
\bibinfo{author}{\bibfnamefont{S.}~\bibnamefont{Boixo}},
  \bibinfo{author}{\bibfnamefont{L.}~\bibnamefont{Aolita}},
  \bibinfo{author}{\bibfnamefont{D.}~\bibnamefont{Cavalcanti}},
  \bibinfo{author}{\bibfnamefont{K.}~\bibnamefont{Modi}}, \bibnamefont{and}
  \bibinfo{author}{\bibfnamefont{A.}~\bibnamefont{Winter}},
  \bibinfo{journal}{International Journal of Quantum Information}
  \textbf{\bibinfo{volume}{9}}, \bibinfo{pages}{1643} (\bibinfo{year}{2011}).

\bibitem[{\citenamefont{Chuan et~al.}(2012)\citenamefont{Chuan, Maillard, Modi,
  Paterek, Paternostro, and Piani}}]{chuan2012quantum}
\bibinfo{author}{\bibfnamefont{T.}~\bibnamefont{Chuan}},
  \bibinfo{author}{\bibfnamefont{J.}~\bibnamefont{Maillard}},
  \bibinfo{author}{\bibfnamefont{K.}~\bibnamefont{Modi}},
  \bibinfo{author}{\bibfnamefont{T.}~\bibnamefont{Paterek}},
  \bibinfo{author}{\bibfnamefont{M.}~\bibnamefont{Paternostro}},
  \bibnamefont{and} \bibinfo{author}{\bibfnamefont{M.}~\bibnamefont{Piani}},
  \bibinfo{journal}{Phys. Rev. Lett.} \textbf{\bibinfo{volume}{109}},
  \bibinfo{pages}{070501} (\bibinfo{year}{2012}).

\bibitem[{\citenamefont{Streltsov et~al.}(2012)\citenamefont{Streltsov,
  Kampermann, and Bru{\ss}}}]{streltsov2012quantum}
\bibinfo{author}{\bibfnamefont{A.}~\bibnamefont{Streltsov}},
  \bibinfo{author}{\bibfnamefont{H.}~\bibnamefont{Kampermann}},
  \bibnamefont{and} \bibinfo{author}{\bibfnamefont{D.}~\bibnamefont{Bru{\ss}}},
  \bibinfo{journal}{Phys. Rev. Lett.} \textbf{\bibinfo{volume}{108}},
  \bibinfo{pages}{250501} (\bibinfo{year}{2012}).

\bibitem[{\citenamefont{Girolami et~al.}(2014)\citenamefont{Girolami, Souza,
  Giovannetti, Tufarelli, Filgueiras, Sarthour, Soares-Pinto, Oliveira, and
  Adesso}}]{girolami2014quantum}
\bibinfo{author}{\bibfnamefont{D.}~\bibnamefont{Girolami}},
  \bibinfo{author}{\bibfnamefont{A.~M.} \bibnamefont{Souza}},
  \bibinfo{author}{\bibfnamefont{V.}~\bibnamefont{Giovannetti}},
  \bibinfo{author}{\bibfnamefont{T.}~\bibnamefont{Tufarelli}},
  \bibinfo{author}{\bibfnamefont{J.~G.} \bibnamefont{Filgueiras}},
  \bibinfo{author}{\bibfnamefont{R.~S.} \bibnamefont{Sarthour}},
  \bibinfo{author}{\bibfnamefont{D.~O.} \bibnamefont{Soares-Pinto}},
  \bibinfo{author}{\bibfnamefont{I.~S.} \bibnamefont{Oliveira}},
  \bibnamefont{and} \bibinfo{author}{\bibfnamefont{G.}~\bibnamefont{Adesso}},
  \bibinfo{journal}{Physical Review Letters} \textbf{\bibinfo{volume}{112}},
  \bibinfo{pages}{210401} (\bibinfo{year}{2014}).

\bibitem[{\citenamefont{Pirandola}(2014)}]{pirandola2014quantum}
\bibinfo{author}{\bibfnamefont{S.}~\bibnamefont{Pirandola}},
  \bibinfo{journal}{Scientific reports} \textbf{\bibinfo{volume}{4}}
  (\bibinfo{year}{2014}).

\bibitem[{\citenamefont{Luo}(2008)}]{luo2008using}
\bibinfo{author}{\bibfnamefont{S.}~\bibnamefont{Luo}}, \bibinfo{journal}{Phys.
  Rev. A} \textbf{\bibinfo{volume}{77}}, \bibinfo{pages}{022301}
  (\bibinfo{year}{2008}).

\bibitem[{\citenamefont{Horodecki and
  Horodecki}(1999)}]{horodecki1999reduction}
\bibinfo{author}{\bibfnamefont{M.}~\bibnamefont{Horodecki}} \bibnamefont{and}
  \bibinfo{author}{\bibfnamefont{P.}~\bibnamefont{Horodecki}},
  \bibinfo{journal}{Physical Review A} \textbf{\bibinfo{volume}{59}},
  \bibinfo{pages}{4206} (\bibinfo{year}{1999}).

\bibitem[{\citenamefont{Brand{\~a}o et~al.}(2015)\citenamefont{Brand{\~a}o,
  Piani, and Horodecki}}]{bph}
\bibinfo{author}{\bibfnamefont{F.~G.} \bibnamefont{Brand{\~a}o}},
  \bibinfo{author}{\bibfnamefont{M.}~\bibnamefont{Piani}}, \bibnamefont{and}
  \bibinfo{author}{\bibfnamefont{P.}~\bibnamefont{Horodecki}},
  \bibinfo{journal}{Nature communications} \textbf{\bibinfo{volume}{6}},
  \bibinfo{pages}{7908} (\bibinfo{year}{2015}).

\bibitem[{\citenamefont{Uhlmann}(1976)}]{uhlmann1976transition}
\bibinfo{author}{\bibfnamefont{A.}~\bibnamefont{Uhlmann}},
  \bibinfo{journal}{Reports on Mathematical Physics}
  \textbf{\bibinfo{volume}{9}}, \bibinfo{pages}{273} (\bibinfo{year}{1976}).

\end{thebibliography}
\end{document}